\documentclass[12pt]{article}
\usepackage{amsmath,amssymb,amsfonts,amsthm}
\usepackage{graphicx}
\usepackage{hyperref}
\usepackage{tikz}
\usetikzlibrary{positioning}
\usepackage{pgfplots}
\usepackage{booktabs}
\usepackage{array}
\usepackage{xcolor}
\usepackage{bm}
\usepackage{geometry}
\pgfplotsset{compat=1.18}

\DeclareMathOperator{\Tr}{Tr}

\newcommand{\I}{\mathbb{I}}
\newcommand{\Pcal}{\mathcal{P}}
\newcommand{\Acal}{\mathcal{A}}

\newtheorem{proposition}{Proposition}
\newtheorem{remark}{Remark}

\begin{document}

\title{
Source-Aware Recovery and Security Bounds in Temporal-Mode Quantum Communication
}

\author{
Ali Vahedi$^{1}$ and Parsa Mahdavifar$^{1}$\\
Department of Physics, Kharazmi University, Tehran, Iran
}

\date{}
\maketitle

\begin{abstract}
We analyze a temporal-mode extension of the two-way LM05 quantum communication primitive, in which a message block is encoded into a fixed-weight binary occupation vector (the support) alongside classical ordering information. We demonstrate that applying an independent-slot recovery model to this correlated source is fundamentally inconsistent. For an ideal erasure channel, we derive the exact maximum-a-posteriori (MAP) recovery probability, showing that source-aware decoding outperforms independent guessing by orders of magnitude. We extend this analysis to a realistic threshold-detector model, quantifying the degradation caused by dark counts and inefficiency. Furthermore, we clarify the quantum-classical information split: in the ideal single-photon limit, preparation averaging renders Eve's state independent of the support, yielding zero Holevo information. We establish conditional security bounds using entropic uncertainty relations and identify photon-number-splitting as the primary practical vulnerability. These results provide a rigorous source-aware recovery benchmark and define the physical conditions under which the quantum layer offers a key-consumption advantage over classical constant-weight coding.
\end{abstract}

\section{Introduction}
\label{sec:intro}

Temporal-mode and time-bin encoding provide a natural way of mapping finite
alphabets onto optical degrees of freedom.  Temporal modes are now a mature
framework for quantum-optical state representation and
processing~\cite{Raymer2020}, with experimental progress in high-dimensional
temporal-mode sorting~\cite{Serino2025}, time-bin quantum
communication~\cite{Yu2025}, and temporal-mode multiplexing for quantum
networks~\cite{Xavier2025,Montaut2025}.  Such approaches are attractive
because a sequence of temporal modes can carry structured classical
information while retaining the physical operations of a qubit-based optical
protocol~\cite{Brecht2015,Fang2018,Xing2013}.  Two-way quantum communication
is particularly interesting in this respect because a single-mode operation
can be applied across a temporal block.

Here we consider a symbol-oriented construction based on the LM05
primitive~\cite{Lucamarini2005}.  An alphabet of size $N$ is mapped to $N$
temporal modes.  For a message containing $K$ distinct symbols, Bob applies
$iY$ to the corresponding $K$ modes and the identity to the remaining modes.
The quantum transmission therefore identifies an unordered set of symbol
positions; a classical description supplies the ordering needed to reconstruct
the message.  The security of LM05 and related two-way protocols depends
critically on source, channel, and attack
assumptions~\cite{Beaudry2013,Henao2015,Patra2025,Patra2026}; extending such
a primitive to multiple temporal modes therefore cannot be justified by a
tensor-product description alone.

The physical motivation for the quantum layer is as follows.  In the ideal
single-photon LM05 protocol, Eve's reduced state after Bob's operation is
$(\I/2)^{\otimes N}$ regardless of the support $P$, because the preparation
average over Alice's four random states erases all support-dependent
information (Sec.~\ref{sec:holevo}).  The $\log_2\binom{N}{K}$ bits of
support information are therefore protected by quantum uncertainty rather than
by a pre-shared key.  A classical constant-weight code sent over the same
fiber would require a one-time-pad key of equal length to achieve the same
confidentiality.  This is the genuine quantum resource exploited by the
construction.

However, this advantage is conditional on the single-photon, lossless,
side-channel-free idealization.  With loss, multi-photon pulses, or detector
imperfections, Eve's Holevo information rises above zero, and a
protocol-specific security proof is required to bound the leakage.  The
present paper does not provide such a proof; it provides the recovery
benchmark and the information-theoretic framework within which such a proof
would be formulated.

The tensor-product operation follows naturally from single-mode operations.
The important issue is instead the \emph{source model}.  Once a block is
constrained to contain exactly $K$ occupied positions, its binary occupation
variables are correlated.  Consequently, a recovery calculation that assigns
independent binary uncertainty to every lost mode does not describe the
stated message ensemble.

This distinction connects the present problem to the established literature
on constant-weight codes and support recovery.  Constant-weight codes have
been studied extensively in coding theory, including their decoding and
erasure-related constructions~\cite{Gadouleau2017,Sasidharan2024}.  Recent
work also develops explicit low-complexity constant-weight constructions and
maximum-likelihood decoders~\cite{Sasidharan2024ML}, motivated in part by
applications in DNA-based data storage~\cite{Erlich2017} and racetrack
memory~\cite{Dau2022}.  Our source contains the full set of $\binom{N}{K}$
weight-$K$ supports rather than a designed subcode, and the MAP expression
follows from posterior symmetry under erasures.

The central result is a benchmark that exposes the potentially large
difference between source-aware decoding and independent-slot guessing.  We
extend the analysis by incorporating a physically motivated detector model
that includes dark counts and detection inefficiency, and we derive the
modified MAP recovery probability.  We also analyze the quantum
information-theoretic aspects, deriving corrected bounds on Eve's information
about the symbol set and quantifying the limitations of protecting only the
classical ordering.

The remainder of the paper first defines the source and ordering model
(Sec.~\ref{sec:protocol}), then derives the exact ideal-erasure MAP benchmark
and compares it with the independent-slot model (Sec.~\ref{sec:loss}).  We
introduce a realistic detector model with dark counts and inefficiency
(Sec.~\ref{sec:detector}), analyze the quantum information-theoretic security
bounds (Sec.~\ref{sec:quantum_info}), and discuss quantum state
discrimination aspects.  The security implications are discussed in
Sec.~\ref{sec:security}, while final sections cover resource scaling
(Sec.~\ref{sec:resources}), implementation constraints
(Sec.~\ref{sec:implementation}), and discussion with a classical comparison
(Sec.~\ref{sec:discussion}).

\section{Temporal-Mode LM05 Construction}
\label{sec:protocol}

\subsection{Encoding}

Let $\Acal$ be an alphabet of size $N$ and let
\[
f:\Acal\rightarrow\{1,\ldots,N\}
\]
be a bijection.  For a message
\[
M=(s_1,\ldots,s_L),
\]
define its unordered symbol-position set by
\begin{equation}
\Pcal_M=\{f(s_1),\ldots,f(s_L)\},
\label{eq:PM}
\end{equation}
and let
\begin{equation}
K=\lvert \Pcal_M\rvert \leq \min(L,N)
\end{equation}
be the number of distinct symbols.

The ordering is represented by
\begin{equation}
\pi:\{1,\ldots,L\}\rightarrow\Pcal_M,
\qquad
\pi(j)=f(s_j).
\label{eq:pi}
\end{equation}

Alice prepares one qubit in each of $N$ temporal slots, using states from
\[
\{|0\rangle,|1\rangle,|+\rangle,|-\rangle\}
\]
and recording the preparation basis.  Bob applies
\begin{equation}
U_i=
\begin{cases}
iY, & i\in\Pcal_M,\\
\I, & i\notin\Pcal_M,
\end{cases}
\qquad
iY=
\begin{pmatrix}
0 & 1\\
-1 & 0
\end{pmatrix}.
\label{eq:Ui}
\end{equation}
For the chosen states, $iY|\psi\rangle$ is orthogonal to $|\psi\rangle$.

The block operation is
\begin{equation}
U_{\mathrm{total}}
=
\bigotimes_{i=1}^{N}U_i.
\label{eq:total_operation}
\end{equation}

Bob returns the temporal sequence to Alice, who measures in her preparation
bases.  In the ideal noiseless case, a surviving photon identifies whether
the corresponding operation was $\I$ or $iY$.  Bob then communicates the
ordering information $\pi$ over an authenticated classical channel.

\begin{figure}[ht]
\centering
\begin{tikzpicture}[node distance=7mm and 9mm,>=stealth]
\node[draw,rounded corners,align=center] (src)
  {Message block $M=(s_1,\ldots,s_L)$};
\node[draw,rounded corners,align=center,right=of src] (set)
  {Fixed-weight support $P=\Pcal_M$, $\lvert P\rvert=K$};
\node[draw,rounded corners,align=center,right=of set] (q)
  {Temporal-mode quantum block $\bigotimes_i U_i$};
\node[draw,rounded corners,align=center,below=of set] (ord)
  {Classical ordering $\pi$};
\node[draw,rounded corners,align=center,right=of q] (rec)
  {Source-aware MAP recovery};
\draw[->] (src) -- (set);
\draw[->] (set) -- (q);
\draw[->] (q) -- (rec);
\draw[->] (src) -- (ord);
\draw[->] (ord) -- (rec);
\end{tikzpicture}
\caption{Logical decomposition of the construction.  The quantum channel
carries the fixed-weight symbol support, while the classical channel carries
ordering information.}
\label{fig:architecture}
\end{figure}
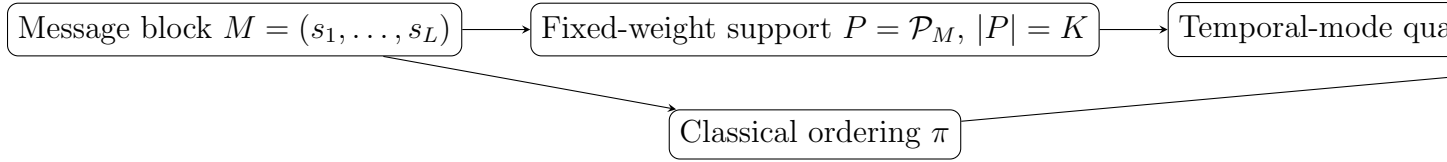

\subsection{The fixed-weight source}

For the recovery analysis, define the binary occupation vector
\begin{equation}
X^N=(X_1,\ldots,X_N)\in\{0,1\}^N,
\qquad
X_i=\mathbf{1}_{\{i\in\Pcal_M\}}.
\end{equation}
For a fixed-$K$ ensemble,
\begin{equation}
\sum_{i=1}^{N}X_i=K.
\label{eq:fixed_weight}
\end{equation}
Under a uniform prior over all $K$-subsets,
\begin{equation}
\Pr(X^N=x)=\binom{N}{K}^{-1}
\quad\text{for }|x|=K,
\end{equation}
and hence
\begin{equation}
H(X^N)=\log_2\binom{N}{K}.
\label{eq:set_entropy}
\end{equation}
The fixed-weight condition is the key structural feature of the recovery
problem.  The coordinates of $X^N$ are not independent.

\begin{table}[ht]
\centering
\caption{Assumptions used by the ideal MAP benchmark and their physical
status.}
\begin{tabular}{p{0.34\linewidth}p{0.16\linewidth}p{0.38\linewidth}}
\toprule
Assumption & MAP model & Physical status \\
\midrule
Uniform fixed-$K$ source & Yes & Idealized benchmark prior \\
Known $K$ at the decoder & Yes & May require side information \\
Independent per-slot survival & Yes & Can fail with detector dead time
  or correlations \\
Perfect identification of erased slots & Yes & Idealized; real
  click/no-click data are noisy \\
No dark counts or false positives & Yes & Idealized \\
No timing jitter or inter-slot confusion & Yes & Idealized \\
Perfect single-photon preparation/measurement & Yes & Requires an explicit
  source/detector model \\
No adversarial parameter estimation & Yes & Security analysis remains
  incomplete \\
\bottomrule
\end{tabular}
\label{tab:assumptions}
\end{table}

\subsection{Classical ordering}

The number of classical bits needed to communicate the ordering depends on
the side information available to Alice.  If the exact multiplicities $m_s$
are known, the number of distinct orderings is
\begin{equation}
|\Pi(M)|=\frac{L!}{\prod_s m_s!},
\end{equation}
and an ideal fixed-length representation requires
\begin{equation}
C_{\mathrm{ord}}
=
\left\lceil
\log_2\frac{L!}{\prod_s m_s!}
\right\rceil.
\label{eq:ordering_cost}
\end{equation}

If Alice knows only the unordered multiplicity pattern, with $c_a$ symbols
having multiplicity $a$, then
\begin{equation}
|\Pi(M)|
=
\frac{K!}{\prod_a c_a!}
\frac{L!}{\prod_s m_s!}.
\label{eq:ordering_case2}
\end{equation}

If she knows only the distinct set, then
\begin{equation}
|\Pi(M)|=K!\,S(L,K),
\label{eq:ordering_case3}
\end{equation}
where $S(L,K)$ is a Stirling number of the second kind.

\subsection{Example: ``Hi Alice''}

For the 53-symbol alphabet consisting of uppercase letters, lowercase
letters, and space, the message ``Hi Alice'' has $L=8$ and $K=7$.  Its
distinct symbol positions are
\begin{equation}
\Pcal_M=\{1,8,29,31,35,38,53\}.
\label{eq:hi_set}
\end{equation}
The symbol $i$ occurs twice.  If Alice knows this multiplicity information,
the ordering cost is
\begin{equation}
C_{\mathrm{ord}}
=
\left\lceil\log_2\frac{8!}{2!}\right\rceil
=15\ \mathrm{bits}.
\label{eq:hi_order}
\end{equation}

The mapping can be made explicit.  Using indices $1$--$26$ for uppercase
letters, $27$--$52$ for lowercase letters, and $53$ for space, the eight
characters map as follows:
\begin{center}
\begin{tabular}{c|cccccccc}
Position & 1 & 2 & 3 & 4 & 5 & 6 & 7 & 8\\
\hline
Symbol & H & i & space & A & l & i & c & e\\
Index & 8 & 35 & 53 & 1 & 38 & 35 & 29 & 31
\end{tabular}
\end{center}
Thus the occupied set is
\[
\Pcal_M=\{1,8,29,31,35,38,53\},
\]
while the ordered message is represented by the rank sequence
\[
(2,5,7,1,6,5,3,4)
\]
after sorting the seven distinct symbols.

There are three useful levels of side information.  In Case~1, Alice already
knows that $i$ is the repeated symbol.  Then there are $8!/2!=20{,}160$
distinct orderings, so $C_{\mathrm{ord}}=15$ bits.  In Case~2, she knows
only that one of the seven symbols is repeated once.  There are $7$ possible
choices for the repeated symbol and therefore
\[
7(8!/2!)=141{,}120
\]
admissible orderings, requiring $18$ bits.  Case~3, in which only the
distinct set is known, gives the same count here because
\[
7!\,S(8,7)=5040\times 28=141{,}120,
\]
again requiring $18$ bits.

For the full seven-distinct-plus-one-repeat ensemble, the number of possible
messages is
\begin{equation}
|\mathcal M|=53\binom{52}{6}\frac{8!}{2!}.
\end{equation}
Hence
\begin{equation}
H(M)\approx44.31\ \mathrm{bits}.
\label{eq:hi_entropy}
\end{equation}
The $15$-bit ordering cost is conditional on the exact multiplicity
information being available to Alice; it is not the entropy of the complete
constrained ensemble.

\section{Exact MAP Recovery for the Fixed-Weight Source}
\label{sec:loss}

\subsection{Ideal erasure model}

Let
\begin{equation}
s=\eta\, T(d)^2
\label{eq:survival}
\end{equation}
be the probability that a prepared photon survives both optical passages and
is detected.  We first consider an ideal erasure model in which Alice knows
which temporal slots were erased.

Suppose $r$ occupied positions and $\ell$ unoccupied positions are erased.
The probability of this event is
\begin{equation}
\Pr(r,\ell)
=
\binom{K}{r}
\binom{N-K}{\ell}
s^{N-r-\ell}(1-s)^{r+\ell}.
\label{eq:erasure_pattern}
\end{equation}

Conditioned on the erased set, exactly $r$ of its $r+\ell$ positions are
occupied.  Under the uniform fixed-weight prior, all $\binom{r+\ell}{r}$
possibilities are equally likely.  Thus the MAP decoder succeeds with
probability
\begin{equation}
\frac{1}{\binom{r+\ell}{r}}.
\end{equation}

We therefore obtain the exact complete-set recovery probability
\begin{equation}
P_{\mathrm{MAP}}(N,K,s)
=
\sum_{r=0}^{K}
\sum_{\ell=0}^{N-K}
\binom{K}{r}
\binom{N-K}{\ell}
s^{N-r-\ell}(1-s)^{r+\ell}
\frac{1}{\binom{r+\ell}{r}}.
\label{eq:exact_map}
\end{equation}

\begin{proposition}
Equation~\eqref{eq:exact_map} is the optimal complete-set recovery
probability for the uniform fixed-weight source over the ideal erasure
channel.
\end{proposition}

\begin{proof}
For each erasure pattern, the receiver knows the unerased coordinates and the
total weight $K$.  The remaining uncertainty consists exactly of choosing
which $r$ of the $r+\ell$ erased coordinates are occupied.  Uniformity of the
prior makes these possibilities equiprobable, so the MAP success probability
is their reciprocal.  Averaging over the erasure patterns gives
Eq.~\eqref{eq:exact_map}.
\end{proof}

\subsection{Independent-slot benchmark}

If the fixed-weight constraint is ignored, an erased slot is assigned an
independent $1/2$ guessing probability.  This gives
\begin{equation}
P_{\mathrm{ind}}(N,s)
=
\left(\frac{1+s}{2}\right)^N.
\label{eq:pind}
\end{equation}

This expression is useful as a benchmark, but it does not represent the
recovery probability of the fixed-weight source.  The difference is
structural: the independent model treats every erased coordinate as a new
binary uncertainty, whereas the fixed-weight decoder uses the global
constraint $\sum_i X_i=K$.

\subsection{Limiting cases}

For $K=1$,
\begin{equation}
P_{\mathrm{MAP}}(N,1,s)
=
s+\frac{1-s^N}{N}.
\label{eq:k1}
\end{equation}

By complement symmetry,
\begin{equation}
P_{\mathrm{MAP}}(N,N-1,s)
=
P_{\mathrm{MAP}}(N,1,s).
\label{eq:complement}
\end{equation}

For $K=N$,
\begin{equation}
P_{\mathrm{MAP}}(N,N,s)=1.
\label{eq:kn}
\end{equation}

For $1\leq K\leq N-1$,
\begin{equation}
\lim_{s\rightarrow0}P_{\mathrm{MAP}}(N,K,s)
=
\binom{N}{K}^{-1}.
\label{eq:large_loss}
\end{equation}

For $s=1-\epsilon$ with $\epsilon\ll1$,
\begin{equation}
P_{\mathrm{MAP}}(N,K,1-\epsilon)
=
1-\frac{K(N-K)}{2}\epsilon^2+O(\epsilon^3).
\label{eq:small_loss}
\end{equation}
The absence of a first-order term follows from the fixed-weight constraint:
a single erased coordinate does not create an occupied-versus-unoccupied
ambiguity.

\subsection{Numerical comparison}

For $N=53$, Fig.~\ref{fig:pmap} compares three fixed-weight ensembles.

\begin{figure}[ht]
\centering
\begin{tikzpicture}
\begin{axis}[
width=0.92\linewidth,
height=7cm,
xlabel={$s$ (per-mode survival probability)},
ylabel={$P_{\mathrm{MAP}}$},
xmin=0.5,xmax=1.0,
ymin=0,ymax=1.02,
grid=both,
legend style={at={(0.03,0.97)},anchor=north west,font=\small}
]
\addplot+[mark=*,samples=100,domain=0.5:1.0]
  {x+(1-x^53)/53};
\addlegendentry{$K=1$ (same as $K=52$)}
\addplot+[mark=*] coordinates {
  (0.5,0.0108616)
  (0.6,0.0366725)
  (0.7,0.1035183)
  (0.8,0.2557157)
  (0.9,0.5686225)
  (0.95,0.8087908)
  (0.99,0.9863558)
  (1,1)
};
\addlegendentry{$K=7$}
\addplot+[mark=*] coordinates {
  (0.5,0.00000219)
  (0.6,0.00006437)
  (0.7,0.00142156)
  (0.8,0.0227569)
  (0.9,0.2412362)
  (0.95,0.6101799)
  (0.99,0.9703454)
  (1,1)
};
\addlegendentry{$K=26$}
\end{axis}
\end{tikzpicture}
\caption{Fixed-weight MAP recovery versus survival probability for $N=53$.}
\label{fig:pmap}
\end{figure}
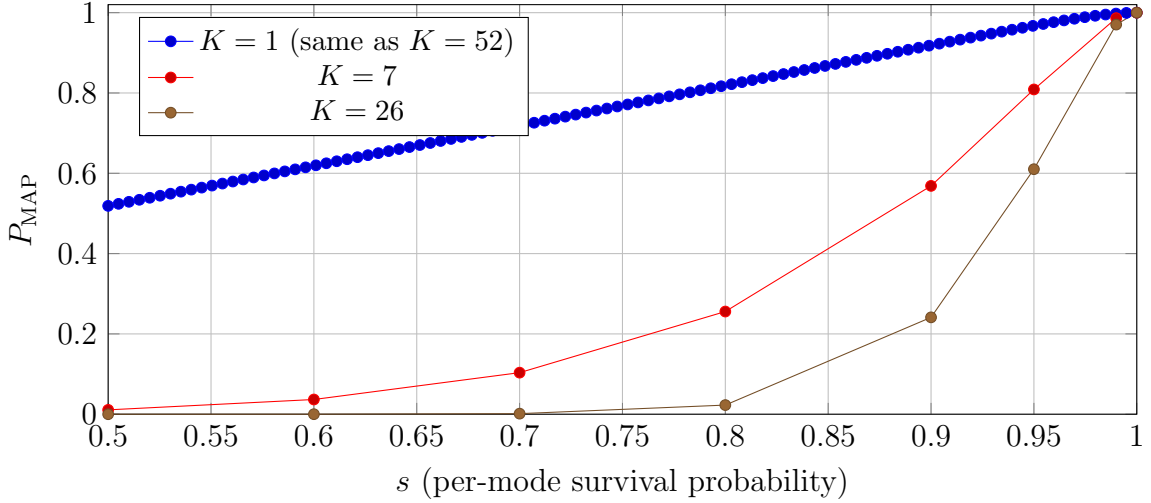

For $N=53$ and $K=7$, representative values are given in
Table~\ref{tab:map_compare}.

\begin{table}[h]
\centering
\caption{Ideal-erasure MAP recovery for $N=53$, $K=7$.}
\begin{tabular}{ccc}
\toprule
$s$ & $P_{\mathrm{MAP}}$ & $P_{\mathrm{MAP}}/P_{\mathrm{ind}}$\\
\midrule
0.50 & $1.086\times10^{-2}$ & $4.55\times10^{4}$\\
0.70 & $1.035\times10^{-1}$ & $5.70\times10^{2}$\\
0.90 & $5.686\times10^{-1}$ & $8.62$\\
0.95 & $8.088\times10^{-1}$ & $3.26$\\
0.99 & $9.864\times10^{-1}$ & $1.29$\\
\bottomrule
\end{tabular}
\label{tab:map_compare}
\end{table}

Thus the independent-slot benchmark can underestimate complete-set recovery
by several orders of magnitude.  However, this comparison assumes error-free
identification of erased and detected modes.  A physical threshold receiver
produces noisy click/no-click records, so the ideal erasure flag is not
reliably available~\cite{Hadfield2009}.

\subsection{Posterior uncertainty beyond exact recovery}

Complete-set success is a stringent operational metric.  The same erasure
pattern also gives a simple posterior-entropy benchmark.  Conditional on an
erasure pattern with $r$ erased occupied positions and $\ell$ erased
unoccupied positions, the posterior is uniform over $\binom{r+\ell}{r}$
compatible supports.  Therefore
\begin{equation}
H(P\mid Y,E_{\mathrm{ch}},r,\ell)
=
\log_2\binom{r+\ell}{r},
\end{equation}
and averaging over erasure patterns gives
\begin{equation}
H(P\mid Y,E_{\mathrm{ch}})
=
\sum_{r=0}^{K}\sum_{\ell=0}^{N-K}
\Pr(r,\ell)\log_2\binom{r+\ell}{r}.
\label{eq:posterior_entropy}
\end{equation}

For $N=53$, $K=7$, and $s=0.9$, this gives
$H(P\mid Y,E_{\mathrm{ch}})\approx1.524$ bits, compared with the prior set
entropy $H(P)=\log_2\binom{53}{7}\approx27.200$ bits.

\subsection{Sensitivity to uncertainty in $K$}
\label{sec:k_uncertainty}

The MAP benchmark assumes the decoder knows $K$ exactly.  In practice, $K$
may be known only approximately.  We analyse the degradation when the true
weight is $K_0$ but the decoder considers
$K'\in\{K_0-\Delta,\ldots,K_0+\Delta\}$ with a prior $\pi(K')$.

Given $c$ detected clicks, the decoder first selects
\begin{equation}
\hat K
=\arg\max_{K'}\;
\pi(K')\,
\frac{\binom{N-c}{K'-c}(1-s)^{K'-c}}{\binom{N}{K'}},
\label{eq:k_hat}
\end{equation}
then guesses uniformly among the $\binom{N-c}{\hat K-c}$ compatible
$\hat K$-supports.  The exact success probability is
\begin{equation}
P_{\mathrm{MAP}}^{\Delta K}
=
\sum_{c=0}^{K_0}
\binom{K_0}{c}s^{c}(1-s)^{K_0-c}\;
\mathbf{1}[\hat K(c)=K_0]\;
\frac{1}{\binom{N-c}{K_0-c}}.
\label{eq:map_dk}
\end{equation}

\begin{proposition}
For the uniform prior $\pi(K')=1/(2\Delta+1)$ over
$K'\in\{K_0-\Delta,\ldots,K_0+\Delta\}$ with $\Delta=1$, and for
physically relevant survival probabilities ($0.1\lesssim s\lesssim1$), the
decoder selects $\hat K=K_0$ if and only if $c=K_0$ (all occupied modes
survive).  Under these conditions,
\begin{equation}
P_{\mathrm{MAP}}^{\Delta K=1}
=s^{K_0}.
\label{eq:map_dk1}
\end{equation}
\end{proposition}

\begin{proof}
For $c<K_0$, the posterior ratio between $K'=K_0$ and $K'=K_0-1$ is less
than unity for all physically relevant $s$, so $\hat K=K_0-1$ and the
decoder fails.  For $c=K_0$, the combinatorial factor
$\binom{N-K_0}{-1}=0$ eliminates the $K'=K_0-1$ candidate, and
$\hat K=K_0$ with probability one.  The support is then fully determined:
$\binom{N-K_0}{0}=1$.
\end{proof}

Table~\ref{tab:k_uncertain} compares the known-$K$ and uncertain-$K$
benchmarks for $N=53$, $K_0=7$, $K'\in\{6,7,8\}$.

\begin{table}[h]
\centering
\caption{MAP recovery with $K$ known exactly versus $K$ uncertain over
$\{6,7,8\}$ (uniform prior), for $N=53$, $K_0=7$.}
\begin{tabular}{cccc}
\toprule
$s$ & $P_{\mathrm{MAP}}$ (known $K$)
    & $P_{\mathrm{MAP}}^{\Delta K=1}=s^{7}$
    & Degradation \\
\midrule
0.50 & $1.086\times10^{-2}$ & $7.81\times10^{-3}$ & $-28\%$\\
0.70 & $1.035\times10^{-1}$ & $8.24\times10^{-2}$ & $-20\%$\\
0.90 & $5.686\times10^{-1}$ & $4.783\times10^{-1}$ & $-16\%$\\
0.95 & $8.088\times10^{-1}$ & $6.983\times10^{-1}$ & $-14\%$\\
0.99 & $9.864\times10^{-1}$ & $9.320\times10^{-1}$ & $-5.5\%$\\
\bottomrule
\end{tabular}
\label{tab:k_uncertain}
\end{table}

\begin{remark}
The collapse $P_{\mathrm{MAP}}^{\Delta K=1}=s^{K_0}$ is a consequence of the
ideal erasure model, in which unoccupied modes never click.  In the detector
model of Sec.~\ref{sec:detector}, dark counts provide additional information
that partially resolves the $K$-ambiguity.  A practical decoder should either
receive $K$ as authenticated side information or jointly estimate $K$ with
the support.
\end{remark}

\subsection{Distance dependence}

The distance dependence of the ideal-erasure benchmark is derived in
Appendix~\ref{app:distance}.  For $N=53$, $K=7$, $\eta=0.9$, and
$\alpha=0.2\;\mathrm{dB/km}$, solving $P_{\mathrm{MAP}}=0.01$ gives
$d_{1\%}\approx6.5\;\mathrm{km}$.  This value illustrates the parametric
sensitivity of the idealized model to propagation loss; it does not
represent an experimentally achievable range.

\section{Realistic Detector Model with Dark Counts and Inefficiency}
\label{sec:detector}

The ideal erasure model assumes perfect identification of erased modes.  In
practice, a threshold detector produces binary click/no-click outcomes with
dark counts and inefficiency.

\subsection{Detection model}

Consider a single temporal mode.  The detector produces a click with
probability:
\begin{equation}
P(\text{click}) =
\begin{cases}
p_{\mathrm{click}}^{(1)} = \eta_d s + (1-\eta_d s)p_d,
  & \text{if occupied},\\
p_{\mathrm{click}}^{(0)} = p_d,
  & \text{if unoccupied},
\end{cases}
\end{equation}
where $\eta_d$ is the detector efficiency, $p_d$ is the dark count
probability per temporal mode, and $s$ is the survival probability.  We
define
\begin{align}
\alpha_1 &= \eta_d s + (1-\eta_d s)p_d, &
\alpha_0 &= p_d, \\
\beta_1 &= (1-\eta_d s)(1-p_d), &
\beta_0 &= 1-p_d.
\end{align}

\subsection{MAP recovery with detector noise}

Let the observed click pattern be $Y\in\{0,1\}^N$.  For a candidate support
$S$ with $|S|=K$, the likelihood is:
\begin{equation}
P(Y|S) = \prod_{i=1}^N
\begin{cases}
\alpha_1^{Y_i} \beta_1^{1-Y_i}, & i \in S,\\
\alpha_0^{Y_i} \beta_0^{1-Y_i}, & i \notin S.
\end{cases}
\label{eq:detector_likelihood}
\end{equation}

The MAP decoder selects the $K$ coordinates with the largest weights
\begin{equation}
w_i = Y_i \log\frac{\alpha_1}{\alpha_0}
    + (1-Y_i)\log\frac{\beta_1}{\beta_0}.
\end{equation}

\subsection{Exact success probability}
\label{sec:detector_success}

Let $a$ be the number of occupied modes that produce clicks, and $b$ the
number of unoccupied modes that produce clicks.  The probability of observing
$(a,b)$ is:
\begin{equation}
\Pr(a,b) = \binom{K}{a} \alpha_1^a \beta_1^{K-a}
           \binom{N-K}{b} \alpha_0^b \beta_0^{N-K-b}.
\end{equation}

The MAP success probability is
\begin{equation}
P_{\mathrm{MAP}}^{\mathrm{det}}(N,K,\eta_d,p_d,s)
=
\sum_{a=0}^{K} \sum_{b=0}^{N-K}
\Pr(a,b)\;g_{\mathrm{det}}(a,b),
\label{eq:detector_success}
\end{equation}
where $g_{\mathrm{det}}(a,b)$ is the conditional success probability.
Since $\log(\alpha_1/\alpha_0)>\log(\beta_1/\beta_0)$ for physical
parameters, every click coordinate has a strictly larger weight than every
no-click coordinate.  Let $m=a+b$ be the total number of clicks.

\paragraph{Case $m<K$.}
The decoder selects all $m$ clicks and $K-m$ no-click coordinates.
Success requires $b=0$ and the $K-a$ no-click selections include all $K-a$
true occupied no-clicks:
\begin{equation}
g_{\mathrm{det}}(a,b)
=
\frac{\delta_{b,0}}{\binom{N-a}{K-a}},
\qquad m<K.
\label{eq:gdet_lt}
\end{equation}

\paragraph{Case $m=K$.}
The decoder selects all $K$ clicks.  Success requires $a=K$, $b=0$:
\begin{equation}
g_{\mathrm{det}}(a,b)
=
\delta_{a,K}\,\delta_{b,0},
\qquad m=K.
\label{eq:gdet_eq}
\end{equation}

\paragraph{Case $m>K$.}
The decoder selects $K$ of the $m$ clicks uniformly.  Success requires
$a=K$:
\begin{equation}
g_{\mathrm{det}}(a,b)
=
\frac{\delta_{a,K}}{\binom{m}{K}},
\qquad m>K.
\label{eq:gdet_gt}
\end{equation}

\begin{proposition}
In the limit $p_d\to0$, $\eta_d\to1$, the detector model reduces to the
ideal erasure model of Eq.~\eqref{eq:exact_map}.
\end{proposition}

\begin{proof}
With $p_d=0$ and $\eta_d=1$: $\alpha_1=s$, $\alpha_0=0$, $\beta_1=1-s$,
$\beta_0=1$.  Hence $b=0$ always and $m=a$.  For $a<K$,
Eq.~\eqref{eq:gdet_lt} gives $g_{\mathrm{det}}=1/\binom{N-a}{K-a}$.
Setting $r=K-a$ and $\ell=N-K$ in Eq.~\eqref{eq:exact_map} gives
$1/\binom{N-a}{K-a}$, confirming agreement.  For $a=K$,
Eq.~\eqref{eq:gdet_eq} gives $g_{\mathrm{det}}=1$, matching the $r=0$ term.
\end{proof}

\subsection{Numerical results}

For $N=53$, $K=7$, $\eta_d=0.9$, $p_d=10^{-4}$:

\begin{table}[h]
\centering
\caption{Detector-aware MAP recovery for $N=53$, $K=7$, $\eta_d=0.9$,
$p_d=10^{-4}$.}
\begin{tabular}{ccc}
\toprule
$s$ & $P_{\mathrm{MAP}}^{\mathrm{det}}$
    & $P_{\mathrm{MAP}}^{\mathrm{ideal}}$ \\
\midrule
0.50 & $1.085\times10^{-2}$ & $1.086\times10^{-2}$ \\
0.70 & $1.034\times10^{-1}$ & $1.035\times10^{-1}$ \\
0.90 & $5.674\times10^{-1}$ & $5.686\times10^{-1}$ \\
0.95 & $8.082\times10^{-1}$ & $8.088\times10^{-1}$ \\
0.99 & $9.863\times10^{-1}$ & $9.864\times10^{-1}$ \\
\bottomrule
\end{tabular}
\label{tab:detector_compare}
\end{table}

The detector noise introduces only small corrections for $p_d=10^{-4}$
because the dark count probability is much smaller than the survival
probability.

\subsection{Dark-count dominated regime and transition}
\label{sec:dc_transition}

When dark counts dominate ($p_d \gg \eta_d s$), the observation becomes
nearly uniform.  In the limit $p_d \to 1$, the detector provides no
information and the recovery probability approaches the prior:
\begin{equation}
\lim_{p_d \to 1} P_{\mathrm{MAP}}^{\mathrm{det}} = \binom{N}{K}^{-1}.
\end{equation}
For $N=53$, $K=7$: $\binom{53}{7}^{-1}\approx7.56\times10^{-9}$.

Table~\ref{tab:dc_transition} shows the transition as $p_d$ increases, for
$N=53$, $K=7$, $\eta_d=0.9$, $s=0.9$.

\begin{table}[h]
\centering
\caption{Detector-aware MAP recovery probability versus dark count
probability $p_d$, for $N=53$, $K=7$, $\eta_d=0.9$, $s=0.9$.}
\begin{tabular}{ccccc}
\toprule
$p_d$ & $\alpha_1$ & $\alpha_0$ & Expected false
  clicks $\langle b\rangle$ & $P_{\mathrm{MAP}}^{\mathrm{det}}$ \\
\midrule
$10^{-4}$ & $0.810$ & $10^{-4}$ & $0.005$  & $5.67\times10^{-1}$ \\
$10^{-3}$ & $0.810$ & $10^{-3}$ & $0.046$  & $5.64\times10^{-1}$ \\
$10^{-2}$ & $0.812$ & $10^{-2}$ & $0.46$   & $\sim5.2\times10^{-1}$ \\
$10^{-1}$ & $0.829$ & $10^{-1}$ & $4.6$    & $\sim10^{-2}$ \\
$0.45$    & $0.896$ & $0.45$    & $20.7$   & $\sim10^{-8}$ \\
\bottomrule
\end{tabular}
\label{tab:dc_transition}
\end{table}

The transition occurs when the expected number of false clicks
$\langle b\rangle=(N-K)p_d$ becomes comparable to $K$.  For
$p_d\lesssim10^{-2}$, false clicks are rare and the detector model closely
tracks the ideal erasure benchmark.  For $p_d\gtrsim10^{-1}$, false clicks
outnumber the true occupied modes, and the success probability drops by
orders of magnitude.

\begin{remark}
The sharp transition between $p_d=10^{-2}$ and $p_d=10^{-1}$ reflects the
combinatorial structure: once $\langle b\rangle$ exceeds $K$, the decoder
must select $K$ modes from a pool dominated by false positives.  This
underscores the importance of low-dark-count detectors (e.g.\ SNSPDs with
$p_d\sim10^{-6}$--$10^{-4}$ per gate) for exploiting the source-aware
advantage.
\end{remark}

\subsection{Computational complexity of the MAP decoder}
\label{sec:complexity}

The MAP decoder does not enumerate the $\binom{N}{K}$ candidate supports.
It computes $N$ log-likelihood weights $w_i$, sorts them, and selects the
$K$ largest.  The total complexity is $O(N\log N)$, independent of $K$.
For $N=53$, this amounts to fewer than $400$ arithmetic operations,
executable in microseconds on any modern processor.  The decoder is
therefore compatible with real-time operation at megahertz block rates.

The $\binom{53}{7}\approx1.54\times10^{8}$ figure is the size of the
support space, not the decoder complexity.  Full enumeration would be needed
only for computing the posterior distribution
(e.g.\ Eq.~\eqref{eq:posterior_entropy}), which is an offline diagnostic.

\subsection{Detector dead time and inter-mode correlations}
\label{sec:dead_time}

The detector model assumes independent per-mode detection.  This is valid
under one of two architectural assumptions: (i)~a detector array with
independent per-mode detection, or (ii)~a single detector with dead time
$\tau_d$ shorter than the slot spacing $\tau_s$.

If $\tau_d>\tau_s$, detection events become correlated: a click in slot $i$
suppresses detection in subsequent slots within the dead-time window.  The
per-mode likelihood $\Pr(Y\mid S)$ in Eq.~\eqref{eq:detector_likelihood} no
longer factorises, and the MAP decoder must use the joint likelihood
\begin{equation}
P(Y\mid S)
=
\prod_{i=1}^{N}
P(Y_i\mid Y_{i-1},\ldots,Y_1;\,S),
\label{eq:corr_likelihood}
\end{equation}
which can be evaluated by dynamic programming in $O(NK)$ time but no longer
reduces to a simple weight-ranking rule.

For the numerical benchmarks in this paper, we assume either a detector
array or $\tau_d<\tau_s$.  Typical SNSPDs have
$\tau_d\sim10$--$50\;\mathrm{ns}$~\cite{Hadfield2009}; for slot spacings
$\tau_s\geq100\;\mathrm{ns}$, the independent-detection assumption is well
satisfied.

\section{Quantum Information-Theoretic Analysis}
\label{sec:quantum_info}

\subsection{Holevo information with preparation averaging}
\label{sec:holevo}

In the LM05 primitive, Alice prepares each mode in a state drawn uniformly
from $\{|0\rangle,|1\rangle,|+\rangle,|{-}\rangle\}$ and records the basis
privately.  Bob applies $U_i\in\{\I,iY\}$.  Eve does not know Alice's basis
choice.

For a single mode with Bob's operation $U_i$, the state available to Eve
(before Alice's basis revelation) is the preparation average
\begin{equation}
\bar\rho_i(U_i)
=
\frac{1}{4}\sum_{\psi\in\{0,1,+,-\}}
U_i|\psi\rangle\langle\psi|U_i^\dagger.
\end{equation}
Because $iY$ permutes the four preparation states up to global phases,
\begin{equation}
iY|0\rangle=-|1\rangle,\quad
iY|1\rangle=|0\rangle,\quad
iY|+\rangle=|{-}\rangle,\quad
iY|{-}\rangle=-|+\rangle,
\end{equation}
one finds
\begin{equation}
\bar\rho_i(\I)=\bar\rho_i(iY)=\frac{\I}{2}.
\label{eq:avg_state}
\end{equation}
Consequently, for \emph{every} support $S$ with $|S|=K$,
\begin{equation}
\bar\rho_S
=
\bigotimes_{i=1}^{N}\bar\rho_i(U_i)
=
\left(\frac{\I}{2}\right)^{\otimes N},
\label{eq:avg_block}
\end{equation}
and the Holevo information about the support vanishes identically:
\begin{equation}
\chi(P:E)
=
S(\bar\rho)
-
\frac{1}{\binom{N}{K}}\sum_S S(\bar\rho_S)
=
N-N=0.
\label{eq:chi_zero}
\end{equation}

\begin{remark}
Equation~\eqref{eq:chi_zero} is the single-photon security mechanism of
LM05: without knowledge of Alice's preparation basis, Eve's reduced state is
independent of Bob's operation, and therefore of the support.  The claim
$\chi=\log_2\binom{N}{K}$ would hold only if Eve knew every $|\psi_i\rangle$
in advance, which contradicts the protocol.
\end{remark}

The zero-Holevo result is an idealization.  In a physical implementation,
several mechanisms give Eve nonzero information: (i)~loss enables
loss-dependent attacks and photon-number-splitting on multi-photon pulses;
(ii)~weak coherent sources produce multi-photon pulses with probability
$O(\mu^2)$; (iii)~detector side channels leak information through the
classical detection record.  The EUR bound of Sec.~\ref{sec:eur} provides a
framework for bounding this leakage once parameter estimation is specified.

\subsection{Entropic uncertainty relation with quantum memory}
\label{sec:eur}

Security against coherent attacks can be bounded using entropic uncertainty
relations (EURs) with quantum memory~\cite{Berta2010,Coles2017}.  For $N$
independent qubits with complementary measurements $X$ and $Z$:
\begin{equation}
H(X^N|E)+H(Z^N|B)\geq N.
\label{eq:block_eur}
\end{equation}

In the LM05 protocol, Alice measures each returned mode in her preparation
basis.  In the noiseless single-photon limit, the measurement outcome in mode
$i$ is deterministic given Bob's operation, so the support $P$ is a
deterministic function of the measurement record $X^N$, and
\begin{equation}
H(P|E)\;\geq\;H(X^N|E)-H(X^N|P).
\label{eq:dp_ineq}
\end{equation}
In the ideal noiseless case $H(X^N|P)=0$, so $H(P|E)\geq H(X^N|E)$.

If parameter estimation gives a complementary-basis error rate $q_\perp$:
\begin{equation}
H(Z^N|B)\leq N\,h_2(q_\perp),
\end{equation}
and combining with \eqref{eq:block_eur} and \eqref{eq:dp_ineq}:
\begin{equation}
I(P;E)\leq
\max\!\Bigl\{0,\;
\log_2\binom{N}{K}-N\bigl[1-h_2(q_\perp)\bigr]\Bigr\}.
\label{eq:leakage_bound}
\end{equation}

\begin{remark}
Equation~\eqref{eq:leakage_bound} is a \emph{conditional} bound: it holds
only after parameter estimation has established $q_\perp$.  It does not
replace a full security proof, which must also handle finite-key effects,
Eve's quantum memory across blocks, and the specific two-way attack structure
of LM05.
\end{remark}

\subsection{State discrimination and support recovery}
\label{sec:discrimination}

The MAP recovery problem of Sec.~\ref{sec:loss} is equivalent to
minimum-error quantum state discrimination among $M=\binom{N}{K}$
equiprobable signal states $\{\rho_S\}$.  The optimal success probability is
\begin{equation}
P_{\mathrm{succ}}^{\mathrm{opt}}
=
\max_{\{\Pi_S\}}
\frac{1}{M}\sum_{S}\Tr(\rho_S\,\Pi_S),
\label{eq:mes}
\end{equation}
where the maximum is over POVMs $\{\Pi_S\}$ with $\Pi_S\geq0$,
$\sum_S\Pi_S=\I$.  The optimality conditions are the
Yuen--Kennedy--Lax (YKL) conditions~\cite{Yuen1975}; no closed-form solution
exists for general $M>2$.

For orthogonal states (ideal, lossless): $P_{\mathrm{succ}}^{\mathrm{opt}}=1$.
After the erasure channel, the effective states become non-orthogonal.  The
MAP recovery probability $P_{\mathrm{MAP}}(N,K,s)$ of
Eq.~\eqref{eq:exact_map} is the optimal discrimination success probability
for the measurement structure imposed by the erasure observation.

\subsection{State distinguishability and the tensor-product structure}
\label{sec:distinguish}

The block state after Bob's operation is
\begin{equation}
|\Psi_S\rangle=\bigotimes_{i=1}^{N}U_i|\psi_i\rangle,
\qquad |S|=K.
\end{equation}
This is a tensor-product state: there is no entanglement or coherence between
temporal modes.  The support information is encoded entirely in \emph{which}
modes carry $iY|\psi_i\rangle$ versus $|\psi_i\rangle$.

For two distinct supports $S\neq S'$, the overlap is
\begin{equation}
|\langle\Psi_S|\Psi_{S'}\rangle|
=
\prod_{i=1}^{N}|\langle\psi_i|U_i^\dagger U_i'|\psi_i\rangle|.
\label{eq:overlap}
\end{equation}
Since $S\neq S'$, there is at least one mode where $U_i\neq U_i'$, and
$|\langle\psi_i|iY|\psi_i\rangle|=0$ for every LM05 preparation state.
Therefore
\begin{equation}
\langle\Psi_S|\Psi_{S'}\rangle=0
\qquad(S\neq S').
\label{eq:ortho}
\end{equation}
The $\binom{N}{K}$ signal states are mutually orthogonal and perfectly
distinguishable by a mode-by-mode measurement in Alice's preparation bases.

\begin{remark}
The orthogonality in \eqref{eq:ortho} is a property of the
\emph{known-basis} states.  After preparation averaging
(Sec.~\ref{sec:holevo}), Eve's state is $(\I/2)^{\otimes N}$ independent of
$S$, so she cannot exploit this orthogonality.  The distinguishability is
available to Alice (who knows her bases) but not to Eve (who does not).
\end{remark}

\subsection{Physical justification: what the quantum layer provides}
\label{sec:physical_just}

In the ideal single-photon regime, Eq.~\eqref{eq:chi_zero} shows that Eve's
quantum state is independent of the support.  The $\log_2\binom{N}{K}$ bits
of support information are protected by quantum uncertainty (the no-cloning
theorem and the indistinguishability of non-orthogonal preparations), not by
a pre-shared classical key.  A classical constant-weight code transmitted
over the same fiber would require encrypting the entire support with a
one-time pad of $\log_2\binom{N}{K}$ bits.

The security relies critically on the two-way (prepare-and-return)
architecture.  Bob's operation $U_i$ is applied to Alice's quantum state,
which Eve cannot clone or measure without introducing disturbance.

The quantum security advantage is fragile: with channel loss, Eve can perform
loss-dependent attacks; with multi-photon pulses, PNS attacks let Eve copy a
photon; with detector imperfections, side-channel attacks can leak support
information.  In each case, $\chi(P:E)$ rises above zero, and the EUR bound
of Eq.~\eqref{eq:leakage_bound} quantifies the residual security margin.

The MAP analysis of Sec.~\ref{sec:loss} establishes the \emph{recovery
performance} of the legitimate receiver.  The security question (how much
information Eve has) is separate and requires the protocol-specific bounds of
Secs.~\ref{sec:holevo}--\ref{sec:eur}.

\subsection{Finite-key considerations}

For finite block lengths, the smooth min-entropy provides the relevant
security metric.  The fixed-weight constraint affects the min-entropy through
the source structure:
\begin{equation}
H_{\min}(X^N) \ge \log_2 \binom{N}{K} - \mathcal{O}\left(\sqrt{N}\right).
\end{equation}
The finite-key security proof would require bounding
$H_{\min}^\epsilon(X^N|E) \ge H_{\min}(X^N) - \text{leakage}$,
where the leakage depends on the error rates in the complementary basis.

\section{Security Implications}
\label{sec:security}

\subsection{Quantum and classical components}

Let $P=\Pcal_M$ denote the symbol set and let $E$ denote Eve's quantum side
information.  The quantum part has the Markov structure
$M\longrightarrow P\longrightarrow E$, and therefore
\begin{equation}
I(M;E)=I(P;E)
\label{eq:message_leakage}
\end{equation}
for the idealized model.  If the ordering is encrypted with an independent
perfect OTP:
\begin{equation}
I(M;E,C_{\mathrm{OTP}})
=
I(M;E)
=
I(P;E).
\label{eq:otp}
\end{equation}
Thus the OTP protects the classical-ordering confidentiality but does not
hide the quantum symbol set.

\subsection{Per-mode versus block statements}

Consider $N=4$, $K=2$, with the six weight-two strings equally likely, and
define $R_E=X_1\oplus X_2$.  Then $P(R_E=0)=1/3$, $P(R_E=1)=2/3$, so
$H(R_E)=h_2(1/3)\approx0.9183$ bits.  For each coordinate,
$I(X_i;R_E)=0$, while $I(X^4;R_E)=H(R_E)>0$.  This shows why a proof that
controls individual-mode leakage is not, by itself, a proof about an
arbitrary function of the entire block.

\subsection{Security status and open problems}

\begin{table}[h]
\centering
\caption{Security status of the temporal-mode symbol-set construction.}
\begin{tabular}{p{0.30\linewidth}p{0.25\linewidth}p{0.35\linewidth}}
\toprule
Question & Status & Required result \\
\midrule
Quantum symbol-set leakage & Open & Protocol-specific block security proof \\
Collective attacks & Open & Block-level entropy/uncertainty analysis \\
Coherent attacks & Open & EAT, de Finetti/post-selection, or direct proof \\
Finite-key security & Open & Explicit composable finite-key
  analysis~\cite{Tomamichel2011} \\
Classical ordering secrecy & Conditional & Authentication plus
  encryption/OTP \\
Cross-block composition & Open & Security with Eve's quantum memory \\
Standalone QSDC & Not established & Full-message composable secrecy and
  correctness proof \\
\bottomrule
\end{tabular}
\label{tab:security_status}
\end{table}

\subsection{Forward-channel attacks and the two-way structure}
\label{sec:forward}

At the forward-channel stage (Alice$\to$Bob), Bob's operation has not yet
been applied, so the quantum state carries zero information about the support
$P$.  Any unitary $V_F$ that Eve applies, together with any ancilla she
retains, is independent of $P$.

On the backward channel, the preparation average (Sec.~\ref{sec:holevo})
ensures that Eve's state is $(\I/2)^{\otimes N}$ regardless of $P$ in the
ideal single-photon case:
\begin{equation}
\bar\rho_i^{\mathrm{bwd}}
=
\frac{1}{4}\sum_{\psi}
V_B\,U_i|\psi\rangle\langle\psi|U_i^\dagger V_B^\dagger
=
\frac{\I}{2}
\quad\forall\;U_i\in\{\I,iY\}.
\end{equation}

The primary vulnerability is the photon-number-splitting (PNS) attack on
multi-photon pulses.  For a weak coherent source with mean photon number
$\mu$, the photon-number distribution is Poissonian, so the multi-photon
probability per pulse is
\begin{equation}
p_{\mathrm{multi}}
=
1-e^{-\mu}-\mu e^{-\mu}
\approx\frac{\mu^2}{2}
\quad(\mu\ll1),
\label{eq:multi_photon}
\end{equation}
where the approximation uses the Taylor expansion of the Poisson
distribution.  For a Fock (single-photon) source, $p_{\mathrm{multi}}=0$
exactly and the PNS attack is eliminated.  Over a block of $N$ modes:
\begin{equation}
P_{\mathrm{PNS}}
=
1-(1-p_{\mathrm{multi}})^N
\approx\frac{N\mu^2}{2}.
\label{eq:pns_block}
\end{equation}
For $N=53$ and $\mu=0.1$, $P_{\mathrm{PNS}}\approx0.265$.
Countermeasures include true single-photon sources, decoy-state
methods~\cite{Lo2005}, and measurement-device-independent architectures.

\section{Resource and Long-Message Scaling}
\label{sec:resources}

The physical quantum resource per block is $N$ prepared temporal modes.  The
symbol-set entropy is $H(P)=\log_2\binom{N}{K}$.

For long messages, let $M=M^{(1)}\Vert\cdots\Vert M^{(B)}$ be a sequence of
blocks.  The quantum resource scales as $N_{\mathrm{ph}}=BN$ prepared photons
before losses.  Under the Case~1 assumption, the ordering key consumption is
\begin{equation}
K_{\mathrm{OTP,total}}
=
\sum_{b=1}^{B}
\left\lceil
\log_2
\frac{L_b!}{\prod_s m_{b,s}!}
\right\rceil.
\label{eq:otp_total}
\end{equation}

The appropriate baseline is a conventional QKD-plus-encryption
architecture~\cite{Scarani2009}.  For the ``Hi Alice'' ensemble, the full
message entropy is about $44.31$ bits, while the support entropy is
$H(P)=\log_2\binom{53}{7}\approx27.20$ bits.  The quantum support carries a
large fraction of the message entropy.  If this support is not protected by a
block-level security proof, protecting only the classical ordering is
insufficient for full-message secrecy.

\section{Implementation Constraints}
\label{sec:implementation}

The MAP benchmarks abstract away detector details.  A physical implementation
must additionally specify timing jitter, slot spacing, detector dead time,
dark counts, synchronization, and thermal
drift~\cite{Hadfield2009,Agrawal2012}.

For independent timing contributions, a useful engineering budget is
\begin{equation}
\sigma_{\mathrm{tot}}^2
=
\sigma_{\mathrm{src}}^2+
\sigma_{\mathrm{det}}^2+
\sigma_{\mathrm{clk}}^2+
\sigma_{\mathrm{elec}}^2+\cdots.
\label{eq:timing_budget}
\end{equation}

If a single detector receives all temporal modes, its dead time limits the
accepted click rate.  More generally,
$t_{\mathrm{block}}=N(\tau_s+\tau_g)+t_{\mathrm{reset}}$, and
$R_{\mathrm{block}}=1/t_{\mathrm{block}}$.

Fiber delay depends on temperature:
\begin{equation}
\frac{d\tau}{dT}
=
\frac{L}{c}
\left(
\frac{dn}{dT}+n\alpha_{\mathrm{th}}
\right).
\label{eq:thermal}
\end{equation}
Static shifts can be calibrated; the relevant experimental requirement is the
residual drift between calibrations together with short-term timing jitter.

\section{Discussion}
\label{sec:discussion}

The central lesson is that the source model must be specified before a
multimode recovery calculation is factorized.  In the present construction,
the relevant object is a uniform constant-weight vector, not $N$ independent
binary variables.  Once that is recognized, the complete-set MAP problem has
the exact solution in Eq.~\eqref{eq:exact_map}.

The resulting discrepancy from independent-slot guessing can be very large.
For $N=53$, $K=7$, and $s=0.5$, the exact MAP probability is approximately
$1.09\times10^{-2}$, whereas the independent benchmark is approximately
$2.39\times10^{-7}$.

The security implications are similarly methodological.  The quantum layer
carries a symbol set and the classical layer carries ordering information.
Protecting the latter does not remove leakage about the former.  A complete
security theorem remains open.

\subsection{Comparison with classical constant-weight alternatives}
\label{sec:classical_compare}

To assess the quantum contribution, we compare the present construction with
a purely classical architecture that transmits the same constant-weight
support over an optical channel.

\paragraph{Classical architecture.}
Alice encodes the support $P$ as a classical binary vector and transmits it
using on/off keying (OOK).  The entire message is encrypted with a one-time
pad.  The support component alone requires
$\log_2\binom{53}{7}\approx27.2$ bits of key.

\paragraph{Quantum architecture (ideal).}
In the ideal single-photon LM05 construction, Eve's state is
$(\I/2)^{\otimes N}$ regardless of the support
(Eq.~\eqref{eq:chi_zero}), so the support is protected by quantum
uncertainty without any pre-shared key.  The key saving is
$\log_2\binom{N}{K}\approx27.2$ bits per block.

\begin{table}[h]
\centering
\caption{Comparison of quantum and classical architectures for the support
component ($N=53$, $K=7$).}
\begin{tabular}{p{0.28\linewidth}p{0.30\linewidth}p{0.30\linewidth}}
\toprule
Metric & Quantum (ideal LM05) & Classical (OOK + OTP) \\
\midrule
Support key cost
  & $0$ bits
  & $\log_2\binom{53}{7}\approx27.2$ bits \\
Ordering key cost
  & $C_{\mathrm{ord}}$ bits
  & included in message OTP \\
Error correction
  & MAP recovery (this paper)
  & Standard FEC or CW code \\
Security mechanism
  & Quantum uncertainty
  & Pre-shared OTP key \\
Loss tolerance
  & Degrades $\chi(P\!:\!E)$
  & Handled by FEC \\
Multi-photon vulnerability
  & PNS attack
  & Not applicable \\
\bottomrule
\end{tabular}
\label{tab:compare}
\end{table}

The quantum architecture offers a genuine reduction in pre-shared-key
consumption for the support, but this advantage erodes with loss,
multi-photon pulses, and detector imperfections.  The classical architecture,
while requiring a larger key, benefits from mature error-correction
technology and does not suffer from PNS attacks.

\section{Conclusion}

We revisited a temporal-mode extension of the two-way LM05 primitive from the
perspective of its actual source ensemble.  The tensor-product construction
maps a message block to a fixed-weight binary occupation vector together with
classical ordering information.  Because the occupation variables are
correlated, the commonly used independent-slot recovery expression does not
represent complete recovery of the stated source.

For a uniform fixed-weight prior and an ideal erasure channel, we derived the
exact MAP recovery probability.  For $N=53$ and $K=7$, this source-aware
result can exceed the independent-slot benchmark by more than four orders of
magnitude at moderate loss.

We extended the analysis to include physically motivated detector noise,
deriving the MAP recovery probability with dark counts and detection
inefficiency.  The source-aware advantage persists for typical detector
parameters but disappears in the dark-count dominated regime.  We analysed
the sensitivity to uncertainty in $K$, showing that a $\pm1$ uncertainty
degrades recovery by $5$--$28\%$, and demonstrated that the MAP decoder runs
in $O(N\log N)$ time, compatible with real-time operation.

The quantum information-theoretic analysis revealed that, in the ideal
single-photon LM05 protocol, preparation averaging over Alice's random basis
choices renders Eve's state independent of the support, yielding zero Holevo
information.  This is the fundamental security mechanism of the construction.
Entropic uncertainty relations provide conditional security bounds once
parameter estimation is specified.  We corrected the state-discrimination
analysis by replacing the inapplicable binary Helstrom bound with the
Yuen--Kennedy--Lax framework for $M$-ary discrimination, and we replaced the
incorrect inter-mode coherence discussion with a rigorous analysis of
tensor-product state distinguishability.

The analysis clarifies the quantum-classical split: an OTP can protect the
ordering information but cannot hide the quantum symbol set.  We showed that
the forward channel carries no support information and that the backward
channel is protected by the LM05 preparation average.  The multi-photon
vulnerability and PNS attack were identified as the primary practical threat.
A comparison with classical constant-weight coding quantified the
$\log_2\binom{N}{K}$-bit key-consumption advantage of the quantum layer in
the ideal regime.

A protocol-specific block security proof would be required before the
construction could support a standalone QSDC claim.  Cross-block composition,
coherent attacks, finite-key effects, and realistic detector models remain
open.

\appendix

\section{Derivation of the MAP expression}

For a fixed-weight source, let $S\subseteq\{1,\ldots,N\}$ be the occupied
set, with $|S|=K$.  Under the ideal erasure channel, each coordinate is
independently retained with probability $s$.  Suppose the erased set contains
$r$ positions from $S$ and $\ell$ positions from its complement.  The
receiver knows the unerased occupied positions and the total weight $K$.
There are $\binom{r+\ell}{r}$ compatible sets, all equally likely under the
uniform prior.  MAP success is the reciprocal.  Averaging over $r$ and
$\ell$ yields Eq.~\eqref{eq:exact_map}.

\section{Message-ensemble entropy}

For fixed symbol multiplicities $m_s$, the number of distinct length-$L$
strings is $|\mathcal M(m)|=L!/\prod_s m_s!$.  For the
seven-distinct-plus-one-repeat ensemble:
\begin{equation}
|\mathcal M|
=
53\binom{52}{6}\frac{8!}{2!},
\end{equation}
which gives approximately $44.31$ bits.

\section{Detector model derivation}

The detector model assumes independent per-mode threshold detection.  The
likelihood for an observation $Y$ given support $S$ is given by
Eq.~\eqref{eq:detector_likelihood}.  The MAP decoder selects the $K$
coordinates with the largest weights.  The success probability is obtained by
summing over all possible observations using
Eqs.~\eqref{eq:gdet_lt}--\eqref{eq:gdet_gt}.

\section{Holevo information with preparation averaging}
\label{app:holevo}

The single-mode preparation average is
$\frac{1}{4}\sum_{\psi}|\psi\rangle\langle\psi|=\I/2$.
Since $iY$ permutes the four preparation states (up to global phases), the
average after $iY$ is also $\I/2$.  Therefore
$\bar\rho_S=(\I/2)^{\otimes N}$ for every $S$, and
$\chi(P:E)=S(\bar\rho)-\frac{1}{\binom{N}{K}}\sum_S S(\bar\rho_S)=N-N=0$.

\section{Distance dependence of the ideal-erasure benchmark}
\label{app:distance}

With fiber attenuation $\alpha=0.2\;\mathrm{dB/km}$ and detector efficiency
$\eta=0.9$, the two-way survival probability is
$s(d)=0.9\times10^{-0.04d}$.  Solving $P_{\mathrm{MAP}}(53,7,s)=0.01$
numerically gives $s\approx0.637$, corresponding to
\begin{equation}
d_{1\%}
=
-\frac{10}{0.4}\log_{10}\!\left(\frac{0.637}{0.9}\right)
\approx6.52\;\mathrm{km}.
\end{equation}
This distance is a property of the idealized erasure model and should not be
interpreted as a protocol range.

\section{Entropic uncertainty relation derivation}

The EUR with quantum memory states $H(X|E)+H(Z|B)\geq\log_2(1/c)$, where
$c=\max_{x,z}|\langle x|z\rangle|^2$.  For qubits, $c=1/2$, giving
$H(X|E)+H(Z|B)\geq1$.  For $N$ qubits: $H(X^N|E)+H(Z^N|B)\geq N$.
Parameter estimation bounds $H(Z^N|B)\leq Nh_2(q_\perp)$, giving
Eq.~\eqref{eq:leakage_bound}.


\end{document}